\documentclass[10pt,twocolumn]{article}
\usepackage[margin=.75in]{geometry}

\usepackage{xcolor}
\usepackage[colorlinks=true,hyperindex=true]{hyperref}
\hypersetup{
    colorlinks,%
    citecolor=violet,%
    linkcolor=violet,%
    urlcolor=violet,%
    pdfnewwindow=true,%
    pdfstartview={FitH}
}

\usepackage{cite}
\usepackage{algorithm,algorithmic}

\usepackage{amsmath,amssymb,amsfonts,amsthm,mathtools}
\newtheorem{theorem}{Theorem}
    
    \newtheorem{lemma}[theorem]{Lemma}

\theoremstyle{definition}
    \newtheorem{definition}[theorem]{Definition}

\newcommand{\Exp}[1]{\mathrm{e}^{#1}}

\newcommand{\nn}{\mathbb{N}}
\newcommand{\cA}{\mathcal{A}}

\newcommand{\Sc}{h^{\textnormal{c}}}
\newcommand{\Sr}{d^{\textnormal{KL}}}

	\makeatletter
	\providecommand*{\diff}%
	{\@ifnextchar^{\DIfF}{\DIfF^{}}}
	\def\DIfF^#1{%
	\mathop{\mathrm{\mathstrut d}}%
	\nolimits^{#1}\gobblespace}
	\def\gobblespace{%
	\futurelet\diffarg\opspace}
	\def\opspace{%
	\let\DiffSpace\!%
	\ifx\diffarg(%
	\let\DiffSpace\relax
	\else
	\ifx\diffarg%
	\let\DiffSpace\relax
	\else
	\ifx\diffarg\{%
	\let\DiffSpace\relax
	\fi\fi\fi\DiffSpace}
\newcommand{\dd}{\diff}

\renewcommand{\P}{\mathbb{P}}
\newcommand{\Q}{\mathbb{Q}}
\newcommand{\X}{\mathbf{X}}
\newcommand{\Y}{\mathbf{Y}}

\newcommand{\PX}{\P_{\X}}
\newcommand{\PY}{\P_{\Y}}

\usepackage{graphicx}

\def\BibTeX{{\rm B\kern-.05em{\sc i\kern-.025em b}\kern-.08em
    T\kern-.1667em\lower.7ex\hbox{E}\kern-.125emX}}

\begin{document}

\title{Ziv--Merhav Estimation for Hidden-Markov Processes}


\author{N.\ Barnfield\footnote{McGill University, Department of Mathematics and Statistics, Montr\'{e}al QC, Canada}\ ,  R.\ Grondin\textsuperscript{\raisebox{-.5ex}{*}}, G.\ Pozzoli\footnote{CY Cergy Paris Universit\'e, Department of Mathematics, Cergy-Pontoise, France} \hspace{.25ex} and R.\ Raqu\'epas\footnote{New York University, Courant Institute of Mathematical Sciences, New York NY, United States}}
\date{}

\maketitle

\paragraph*{Abstract}
    We present a proof of strong consistency of a Ziv--Merhav-type estimator of the cross entropy rate for pairs of hidden-Markov processes. Our proof strategy has two novel aspects: the focus on decoupling properties of the laws and the use of tools from the thermodynamic formalism.\\

    \noindent\textit{Presented at the IEEE International Symposium on Information Theory 2024, in Athens, Greece.}



\section{Introduction}
We are interested in Ziv--Merhav-type estimators~\cite{MZ93} of the cross entropy rate 
for pairs of sources\,---\,which is a sum of a Kullback--Leibler (KL) divergence rate and an entropy rate. While such estimators are widely used in practice (see {\em e.g.}~\cite{KPK01,BCL02,CF05,B+08,CFF10,RP12,LMDEC19,R+22}), theoretical works on the subject are scarce. The goal of this proceedings paper is to concretely present the consequences of recent findings~\cite{BCJPPExamples,CDEJR23w,CR23,BGPR,BGPRb} for hidden-Markov processes in an accessible way.

Throughout, $\cA$ is a finite alphabet and $\cA^\nn$ is the space of one-sided $\cA$-valued sequences, denoted by bold lower-case letters, {\em e.g.}~$\mathbf{x}=(x_n)_{n=1}^\infty$. We use $x_k^\ell$ for the string ({\em i.e.}\ finite concatenation of elements of~$\cA$) of length $\ell-k+1$ starting at index~$k$ in such a sequence~$\bf{x}$. We also use $[x_1^l] := \{\mathbf{z} \in \cA^{\nn}: z_1^l = x_1^l\}$. We use {bold upper-case letters} for $\cA$-valued processes, {\em e.g.}~$\X = (X_n)_{n=1}^{\infty}$. We will only consider processes that are stationary, {\em i.e.}\ processes whose law (a measure on~$\cA^{\nn}$) is invariant under the shift $ (x_{n})_{n=1}^\infty \mapsto (x_{n+1})_{n=1}^\infty$.

\subsection{Entropies}
\label{ssec:entropies}

The \emph{Shannon entropy rate} of a stationary process~$\X$\,---\,or equivalently of its law $\PX$\,---\,is the limit
\[
    h(\PX) = \lim_{n\to\infty} \frac{\sum_{a\in\cA^n} \PX[a] (-\ln\PX[a])}{n}.
\]
There are at least two important related quantities for pairs of stationary processes:
the
\emph{cross entropy rate} and the \emph{KL divergence rate} (a.k.a.\ relative entropy rate): 
\[
    \Sc(\PY|\PX) = \lim_{n\to\infty} \frac{\sum_{a\in\cA^n} \PY[a] (-\ln\PX[a])}{n}
\]
whenever the limit exists in~$[0,\infty]$, and then 
\(
    \Sr(\PY|\PX) = \Sc(\PY|\PX) - h(\PY).
\)
As we shall see in Section~\ref{ssec:dec}, this will be the case for the pairs of measures we are interested in, but it should be noted that this is not a general fact about pairs of ergodic processes.%
    \footnote{\color{black} Let us emphasize that the KL divergence of~$\PY$ with respect to $\PX$ itself\,---\,not the rate\,---\,is well defined but too coarse to be useful in many tasks. Indeed, if both measures are ergodic, then this quantity is either zero or infinite.
    } 

The KL divergence rate plays a fundamental role in many tasks in information theory and statistics: binary hypothesis testing, universal classification, {\em etc}. Because the KL divergence rate differs from the cross entropy rate by an entropy rate which can be universally estimated following the Lempel--Ziv parsing algorithm~\cite{LZ78}, we will focus on estimation of the cross entropy rate.

\subsection{Ziv--Merhav-type estimators}
\label{ssec:ZM}

Inspired by the dictionary-based approach of the Lempel--Ziv algorithm for estimating the entropy rate, researchers have sought to develop similar universal estimators that generalize to multiple sources to measure the cross entropy rate.
Introduced in~\cite{MZ93}, the Ziv--Merhav (ZM) estimator is one such cross entropy rate estimator, whose consistency was established in the case where $\X$ and $\Y$ are stationary multi-level Markov processes; we will denote it by $Q_N^{\text{ZM}}$. 
We consider a modification $Q_N$ of $Q_N^{\text{ZM}}$, introduced recently in \cite{BGPRb}, which we refer to as the \emph{modified Ziv--Merhav} (mZM) \emph{estimator}. 

\begin{definition}
\label{def:main}
    For realizations $(\mathbf{y},\mathbf{x})$ of $(\Y,\X)$, the mZM parsing of $y_1^N$ with respect to $x_1^N$ begins by determining the shortest prefix $\overline{w}^{(1,N)}$ of $y_1^N$ that does not appear in $x_1^N$. 
    Then, $\overline{w}^{(2,N)}$ is the shortest prefix of the unparsed part of~$y_1^N$ that does not appear in $x_1^N$, and so on until we reach the end of~$y_1^N$. 
    The \emph{parsing length} $c_N(\mathbf{y},\mathbf{x})$ is the number of words in the parsing
    \[
        y_1^N = \overline{w}^{(1,N)} \overline{w}^{(2,N)} \dotsc \overline{w}^{(c_N,N)};
    \] see Algorithm~\ref{alg:mZM}.
    The \emph{mZM estimator} is
    \[
    Q_N(\mathbf{y}, \mathbf{x}) := \frac{c_N(\mathbf{y}, \mathbf{x}) \ln N}{N - c_N(\mathbf{y}, \mathbf{x})}.
    \]
\end{definition}

In essence, the original ZM algorithm parses~${y}_1^N$ according to the \emph{longest words found} in ${x}_1^N$ whereas the mZM algorithm parses according to the \emph{shortest words not found}: the difference between the ZM and mZM parsing lengths boils down to the choice of 
  imposing or not the condition ``if $i=j$'' for executing Line~7 in Algorithm~\ref{alg:mZM}.
This difference is compensated by the choice of not subtracting or subtracting $c_N$ at the denominator in the definition of the estimator. As seen in Section~\ref{sec:num}, these subtleties yield no observable difference in performance.
\begin{algorithm}
\caption{Computation of the mZM parsing length}
\label{alg:mZM}
\begin{algorithmic}[1]
    \REQUIRE $x_1^N, y_1^N \in \cA^N$ 
    \ENSURE $c_N$
    \STATE $c_N \gets 1$, $j \gets 1$, $i \gets 1$
    \WHILE{$j < N$}
        \IF{$y_{i}^j$ is in $x_1^N$}
        \STATE $j \gets j + 1$
        \ELSE
            \STATE $c_N \gets c_N + 1$
            \STATE $j \gets j +1$
            \STATE $i \gets j$
        \ENDIF
    \ENDWHILE
\end{algorithmic}
\end{algorithm}

In a way, ZM-type estimators make repeated use of the notion of \emph{longest-match length}
\[
\Lambda_N(\mathbf{z},\mathbf{x}) = \sup \{ l : z_1^l = x_k^{k+l-1} \text{ for some } k \leq N - l + 1 \}
\]
for which, under suitable assumptions on $\Y$ and $\X$, we have 
\[ 
    \lim_{N\to\infty} 
    \frac{\ln N}{\Lambda_N(\Y,\X)}
   = \Sc(\PY|\PX)
\]
almost surely; see {\em e.g.}~\cite{WZ89,OW93,Ko98,CDEJR23w}. However, these repeated uses of longest-match lengths are not independent, making it delicate to try to deduce convergence of ZM-type estimators from convergence of longest-match length estimators.

\subsection{Hidden-Markov processes}
\label{ssec:hmm}

An $\cA$-valued process is called a hidden-Markov process (HMP) if all the marginals of its law can be represented as 
\[ 
    \P[x_1^n] = \sum_{s \in \mathcal{S}^n} \pi_{s_1} R_{s_1,x_1} P_{s_1,s_2} R_{s_2,x_2} \dotsb P_{s_{n-1},s_{n}} R_{s_n,x_n}
\]
for some fixed Markov chain $(\pi,P)$ on some state space~$\mathcal{S}$\,---\,this is the ``hidden'' chain\,---\,and some fixed $(\#\mathcal{S})$-by-$(\#\cA)$ stochastic matrix~$R$. Such processes are known under different names, including \emph{probabilistic functions of Markov processes}, and have gained immense popularity in statistical applications since the seminal papers~\cite{BP66,Pe69}; also see the review~\cite{EM02}.

Throughout this paper, we will only consider $\cA$-valued HMPs that can be represented with the following constraints: 
\begin{enumerate}
    \item[i.] the hidden state space is finite;
    \item[ii.] the hidden chain is stationary and irreducible;
    \item[iii.] for each $s_1 \in \mathcal{S}$, there exists $n$ such that $\sum_{t \in \mathcal{S}^{n-1}} \pi_{s_1} R_{s_1,x_1} P_{s_1,t_1} R_{t_1,x_2} \dotsb P_{t_{n-2},t_{n-1}} R_{t_{n-1},x_n}$ is positive for more than one string $x$ of length~$n$.
\end{enumerate}
Condition~ii implies ergodicity and Condition~iii is essentially the minimal condition for 
the process not to be eventually deterministic. 
As far as the class of measures on~$\cA^\nn$ that can be obtained is concerned, it turns out that there is no loss of generality in considering only deterministic functions of (possibly larger) hidden chains. The Shannon entropy rate of such HMPs has been the subject of many works; see {\em e.g.}~\cite{Bl57,Bi62,JSS04,ZDKA06,HM06}.

\subsection{Main result}

Our main result is the following result on strong consistency of the mZM estimator of the cross entropy rate from Definition~\ref{def:main}.  To our knowledge, no analogue of this theorem is available for the ZM estimator.

\begin{theorem}
\label{thm:ieee-main}
    Suppose that $\X$ and $\Y$ are independent, HMPs with respective laws $\PX$ and $\PY$. If they both satisfy Conditions~i--iii, then
    \[ 
        \lim_{N\to\infty} Q_N(\Y,\X) 
        = \Sc(\PY|\PX)
    \]
    almost surely.
\end{theorem}

Note that Theorem~\ref{thm:ieee-main} provides an estimation of the cross entropy rate between the HMPs~$\X$ and~$\Y$, not the underlying (hidden) Markov processes.
In Section~\ref{sec:prelim}, we discuss some important preliminaries for our presentation of the proof in Section~\ref{sec:proof}. The result is illustrated using numerical experiments in Section~\ref{sec:num}.

\section{Decoupling properties of hidden-Markov processes}
\label{sec:prelim}

In this section, we identify the key properties of HMPs that we will use for the proof of Theorem~\ref{thm:ieee-main}. While these properties are natural from the point of view of the statistical mechanics of lattice gases, we suspect that they might not be familiar to the information theory community.

\subsection{Decoupling inequalities and their first consequence}
\label{ssec:dec}

It is shown in~\cite{BCJPPExamples} that if $\P$ is the law of a HMP satisfying i--ii, then there exist two natural numbers $k$ and $\tau$ with the following two properties: 
\begin{itemize}
    \item for all strings $a$ and $b$, 
    \begin{equation}
    \label{eq:UD}
        {\P[a b]} \leq \Exp{k}{\P[a]\P[b]};
    \end{equation}
    \item for all strings $a$ and $b$, there exists a string $\xi$ with length $|\xi| \leq \tau$ such that 
    \begin{equation}
    \label{eq:SLD}
         {\P[a \xi b]}  \geq \Exp{-k} {\P[a]\P[b]}.
     \end{equation}
\end{itemize}
These properties and generalizations thereof\footnote{A particularly useful generalization consists in allowing $\tau$ and $k$ to grow sublinearly with the length of the string~$a$. However, we will not consider this generalization as it is unnecessary for HMPs and complicates some arguments. Still, as stated here, the inequalities~\eqref{eq:UD}--\eqref{eq:SLD} are weaker than the so-called ``quasi-Bernoulli'' property and imply no form of mixing.} are known as \emph{decoupling} properties and have been instrumental in recent progress in large-deviation theory and its connection to statistical mechanics and information theory; see~\cite{Pf02,CJPS19,CR23}. A particularly convenient tool in their derivation in the context of HMPs is the so-called ``positive-matrix product'' representation; see {\em e.g.}~\cite{BCJPPExamples}.
The following is a straightforward consequence of Kingman's subadditive ergodic theorem~\cite{Ki68}.
\begin{lemma}
\label{lem:cons-KSET}
    If~$\PX$ satisfies~\eqref{eq:UD} and $\PY$ is ergodic, then the limit $\Sc(\PY|\PX)$ exists and
    \begin{equation}
    \label{eq:cross-SMB}
        \lim_{n\to\infty} \frac {-\ln \PX[Y_1^n]}{n} = \Sc(\PY|\PX)
    \end{equation}
    almost surely.
\end{lemma}
It should be noted that the role of~\eqref{eq:UD} for the study of the KL divergence rate for HMPs dates back at least to~\cite{Le92}.

\subsection{The thermodynamic formalism} 
\label{ssec:ent-thermo}

In our proof, a key role is played by the following \emph{pressure}: 
\begin{equation}
\label{eq:qbar-def}
    \bar{q}(\alpha) := \limsup_{\ell\to\infty}\frac{q_\ell(\alpha)}\ell,
\end{equation}
where $\alpha$ is a real variable and
\begin{equation}
\label{eq:qn-def}
    q_\ell(\alpha) := \ln \sum_{a \in \cA^\ell} \Exp{-\alpha\ln \PX[a]} \PY[a].
\end{equation}
The function~$\bar{q}$ is a (possibly improper) monotone, convex function satisfying $\bar{q}(0) = 0$ and $\bar{q}(-1) < 0$, the latter being a consequence of Condition~iii.
While it is not a pressure in the sense of the usual (additive) thermodynamic formalism, the bound~\eqref{eq:UD} provides the requirement for the \textit{subadditive thermodynamic formalism} of {\em e.g.}~\cite{CFH08}. As such, the pressure~$\bar{q}$ still satisfies the variational principle
\begin{equation}
\label{eq:q-var}
    \bar{q}(\alpha) = \sup_{\Q } \left[\alpha \int f_{\PX} \dd\Q  - d^{\textnormal{KL}}(\Q |\PY)\right]
\end{equation}
for all $\alpha \leq 0$, where the supremum is taken over all shift-invariant laws~$\Q $ and 
\[ 
   f_{\PX}(\mathbf{z}) := \limsup_{n\to\infty} \frac{-\ln \PX[z_1^n]}{n}.
\] 
The role of such variational principles in the study of entropic quantities\,---\,which can be traced back to ideas of Gibbs~\cite{Gib}\,---\,has become ubiquitous in the mathematical literature on statistical mechanics, dynamical systems, and large deviations building on seminal works of Ruelle, Bowen and Walters, {\em e.g.}~\cite{Ru73,Bo74,Wa75b}. 
While we have no guarantee that~$\bar{q}$ is differentiable at the origin, the variational principle~\eqref{eq:q-var} allows us to identify the left derivative; see Figure~\ref{fig:q-and-h}.
\begin{lemma}
If $\PX$ and $\PY$ satisfy~\eqref{eq:UD} and are ergodic, then
\label{lem:left-der}
    \begin{equation*}
         D_- \bar{q}(0) = \Sc(\PY|\PX).
    \end{equation*}
\end{lemma}
\begin{figure}
    \centering
    \includegraphics{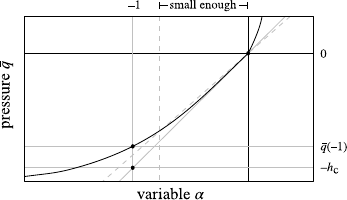}
    \caption{The pressure may or may not be differentiable, but its left derivative at the origin can still be identified as the cross entropy rate under quite general assumptions. As a visual aid to the convexity argument for Lemma~\ref{lem:ing-IV}, the gray dashed line (figuratively) has slope $D_-\bar{q}(0)-\tfrac \epsilon 4$.}
    \label{fig:q-and-h}
\end{figure}
The proof of this nontrivial fact relies on a uniqueness argument for the law maximizing~\eqref{eq:q-var} in $\alpha = 0$ that can be traced back to ideas of Ruelle, and which is explained {\em e.g.}\ in~\cite{BJPP18}; also see~\cite{Sim}. 
For actual Markov chains ({\em i.e.}\ when~$R$ is a permutation matrix), the left derivative does coincide with the right derivative as a consequence of the Perron--Frobenius theorem and standard perturbation theory arguments.

\subsection{Waiting-time estimate}
\label{ssec:KB}

The number of words in ZM-type parsings is intimately related to the so-called Wyner--Ziv problem on waiting times introduced in~\cite{WZ89} and notably studied in~\cite{Sh93,MS95,Ko98}. The key quantity in this problem is the time 
\[ 
    W(a,\mathbf{x}) = \inf\{r : x_{r}^{r+|a|-1} = a\}
\]
we need to wait to see the string~$a$ appear in a sequence~$\mathbf{x}$. Waiting times are dual to the longest-match lengths from Section~\ref{ssec:ZM}~\cite{WZ89}, but typically allow for more direct bounds.
It was recently shown in~\cite{CDEJR23w} that, for every string $a$ and natural number $r$,
\begin{equation}
\label{eq:KB}
    \operatorname{Prob} \{W(a,\X) > r\} \leq (1 - \Exp{-k}\PX[a])^{\left\lfloor\frac{r-1}{|a|+\tau}\right\rfloor},
\end{equation}
provided that the law~$\PX$ of~$\X$ under the underlying measure~$\operatorname{Prob}$ satisfies~\eqref{eq:SLD}.
This progress was made by revisiting the ideas of~\cite{Ko98} from the perspective of decoupling properties. In the recent works~\cite{AACG22,CR23}, a detailed analysis of the pressure $\bar{q}$ plays a crucial role in the study of large deviations of waiting times.

\section{Proof strategy}
\label{sec:proof}

We are now in a position to present the proof strategy for Theorem~\ref{thm:ieee-main}. 
{Still on~$\cA$, the same exact strategy applies to irreducible Kusuoka processes~\cite{JOP17,BCJPPExamples} and to $\psi$-mixing processes that satisfy $\psi^*(0) < \infty$ in the notation of~\cite{Bra05}, as soon as one can check the nondegeneracy condition $\bar{q}(-1) < 0$. 
Further generalizations leveraging relaxations of the decoupling conditions~\eqref{eq:UD}--\eqref{eq:SLD} are discussed in our more technical paper~\cite{BGPRb}. }

Throughout this section, the processes $\X$ and $\Y$ are fixed processes with hidden-Markov representations satisfying Conditions~i--iii from Section~\ref{ssec:hmm} and have respective laws $\PX$ and $\PY$. We will consider them as defined on a common underlying probability space in such a way that they are independent and use ``$\operatorname{Prob}$'' for the probability measure on this space. 
For each~$N$, the number $c_N$ of words, and the words $\overline{w}^{(1,N)}$, $\overline{w}^{(2,N)}$, $\dotsc$, $\overline{w}^{(c_N,N)}$ are functions of~$X_1^N$ and~$Y_1^N$  and can therefore be defined through composition as random variables on that same underlying probability space. 

\begin{definition}
    
    We use $\underline{w}^{(i,N)}$ for the (possibly empty) string obtained by taking $\overline{w}^{(i,N)}$ in Definition~\ref{def:main} without its last letter. 
\end{definition}

We use ``with high probability'' to describe sequences $(E_N)_{N=1}^{\infty}$ of events with the property that $\operatorname{Prob}(E_N) > 1 - O(N^{-2})$ as $N\to\infty$. 
In particular, the Borel--Cantelli (BC) lemma can be applied to the complement of finitely many such events to deduce almost sure statements.

First, one can show that if there is a string $a$ such that $\PY[a] > 0$ but $\PX[a] = 0$, then ergodicity implies that the theorem holds true in the sense that both sides must be infinite. Hence, we will assume from now on that this does not happen. For simplicity, we will present the proof in the case
$D_-\bar{q}(0) < \infty$. In the complementary infinite case, some of the bounds become vacuously true and the others are proved superficially adapting the same sequence of arguments as in the finite case.

\begin{lemma}
\label{lem:a-priori}
    There exists a constant $\kappa > 0$ such that
    \begin{align}
    \label{eq:UB on words}
        \max_{i=1,\dotsc,c_N} |\overline{w}^{(i,N)}|
        < \kappa \ln N
    \end{align}
    with high probability, and for every~$\lambda > 0$, 
    \begin{align}
    \label{eq:LB on words}
        \lambda
        &\leq \min_{i=1,\dotsc,c_N-1} |\underline{w}^{(i,N)}|
    \end{align}
    ---\,and thus $c_N \leq  \lambda^{-1} N$\,---\,with high probability.
\end{lemma}

\begin{proof}[Proof sketch]
    For every $\ell$, by stationarity and a union bound (over the locations at which $Y_1^\ell$ can appear in~$\X$), we have 
    \begin{align}
        \operatorname{Prob}\left\{ W(Y_1^\ell, \X) < N\right\}
            &\leq \sum_{a \in \cA^\ell} \PY[a] \cdot N \PX[a] 
            \notag \\
            &= N \Exp{q_\ell(-1)},
            \label{eq:W-union}
    \end{align}
    where $q_\ell$ is defined in~\eqref{eq:qn-def} and $q_\ell(-1) = \ell \bar{q}(-1) + o(\ell)$ by~\eqref{eq:qbar-def}. Also recall from Section~\ref{ssec:ent-thermo} that $\bar{q}(-1)<0$ as a consequence of Condition~iii.
    In particular, taking $\ell = \kappa \ln N$ with $\kappa$ large enough, the right-hand side of~\eqref{eq:W-union} can be made to decay as fast as any inverse power of~$N$.
    By stationarity, the same is true with $Y_1^\ell$ replaced with $Y_{m}^{m+\ell}$. Hence, by a union bound (over $m$), we find that, with high probability, no substring of $Y_1^N$ of length $\kappa \ln N$ appears in~$X_1^N$.

    On the other hand, the bound~\eqref{eq:KB} gives an exponentially decaying bound on the probability that some $a$ does not appear in~$X_1^N$. But that bound can be made uniform in the choice of~$a$ of \emph{fixed} length $\lambda$ with $\PY[a] > 0$.  Hence, by a union bound (over the choices of~$a$), we find that, with high probability, all possible strings of length $\lambda$ appear somewhere in $X_1^N$. 
\end{proof}

\begin{lemma}
\label{lem:ing-I}
    For every $\epsilon > 0$,
    $$
        (1-\epsilon)[c_N-1] \ln N 
                \leq -\sum_{i=1}^{c_N-1} \ln \PX[\overline{w}^{(i,N)}]
    $$
    with high probability.
\end{lemma}

\begin{proof}[Proof sketch]
    The two crucial observations are the following:
    \begin{itemize}
        \item If a sum of $c_N-1$ nonnegative terms is smaller than $(1-\epsilon)[c_N-1] \ln N$, then at least one term must be smaller than $\ln N^{1-\epsilon}$;
        \item If $\overline{w}^{(i,N)}$ is a word in a parsing of~$Y_1^N$ that uses words not found in~$X_1^N$, then $W(\overline{w}^{(i,N)},\X) \geq N - |\overline{w}^{(i,N)}|$.
    \end{itemize}
    The inequality~\eqref{eq:KB} from Section~\ref{ssec:KB} provides a tension between those two observations: an index~$i$ with small $-\ln\PX[\overline{w}^{(i,N)}]$ is rarely associated with a large waiting time.
    Thanks to Lemma~\ref{lem:a-priori}, these observations can be turned into a rigorous proof using standard probabilistic techniques.
\end{proof}
 It should be noted that the proof of Lemma~\ref{lem:ing-I} is the only step that works for the mZM estimator, but not for the ZM estimator.

\begin{lemma}
\label{lem:ing-III}
   Almost surely,
    \[
        -\sum_{i=1}^{c_N-1} \ln \PX[\overline{w}^{(i,N)}]
        \leq - \ln \PX[Y_1^N] + kc_N.
    \]
\end{lemma}

\begin{proof}[Proof sketch]
    Start with the string~$Y_1^N$ and apply the upper decoupling inequality~\eqref{eq:UD} repeatedly.
\end{proof}

\begin{lemma}
\label{lem:ing-II}
    For every $\epsilon > 0$,
    \begin{equation}
    \label{eq:ing-II}
         -\sum_{i=1}^{c_N} \ln \PX[\underline{w}^{(i,N)}]
            \leq (1+\epsilon)c_N \ln N
    \end{equation}
    with high probability.
\end{lemma}

\begin{proof}[Proof sketch]
    Fix $\epsilon \in (0,1)$ and let $\lambda \geq 2$ be arbitrary for the time being. By Lemma~\ref{lem:a-priori}, it is harmless to restrict our attention to the event that~\eqref{eq:UB on words}--\eqref{eq:LB on words} hold.
    With the intent of using a union bound over the possible choices of values~$c$ for $c_N$, and then the ways of choosing the starting index $L_i$ of each $\underline{w}^{(i,N)}$ for $i=1,\dots ,c$, we seek an upper bound on
    \begin{align*}
    \operatorname{Prob}\left(\left\{-\sum_{i=1}^{c} \ln\PX[Y_{L_i}^{L_{i+1}-2}]\geq (1+\epsilon)c\ln N\right\}\cap\bigcap_{i=1}^{c} B_{i}\right)
    \end{align*}
    where $B_i$ is used as the shorthand 
    \[
        B_{i}:=\left\{W(Y_{L_i}^{L_{i+1}-2},\X)<N\right\}
    \]
    with the convention that $L_1=1$ and $L_{c+1} = N+1$ .
    
    Now, for any $(t_i)_{i=1}^{c}\in\mathbb{R}^c$, we can successively use the upper decoupling inequality~\eqref{eq:UD} for the law~$\PY$ to obtain
    \begin{multline}
        \operatorname{Prob}\left(\bigcap_{i=1}^{c}\left\{-\ln\PX[Y_{L_i}^{L_{i+1}-2}]\geq t_i\right\} \cap B_{i}\right) \\
        \leq \Exp{c k}\prod_{i=1}^{c}\operatorname{Prob}\left(\left\{-\ln\PX[Y_{L_i}^{L_{i+1}-2}]\geq t_i\right\} \cap B_{i}\right). 
        \label{eq:meas}
    \end{multline}
    By a union bound and shift-invariance, we also have that
    \begin{multline*}
       \operatorname{Prob}\left(\left\{-\ln\PX[Y_{L_i}^{L_{i+1}-2}]\geq t_i\right\} \cap B_{i}\right)\\
        \leq \sum_{\substack{a\in\cA^{L_{i+1}-L_i-1},\\ -\ln \PX[a]\geq t_i}}\PY[a]\;N\PX[a]\leq N\Exp{-t_i}.
    \end{multline*}
    This shows that the measure defined by the left-hand side of inequality~\eqref{eq:meas} is bounded above by a product measure on $\mathbb{R}^{c}$, with each marginal having a well-defined moment-generating function on $(0,1)$. Using a Chebyshev-like bound for the function $t\mapsto \Exp{(1+\epsilon)^{-1/2}t}$, we have
    \begin{multline*}
    \operatorname{Prob}\left(\left\{-\sum_{i=1}^{c} \ln\PX[Y_{L_i}^{L_{i+1}-2}]\geq 
    t \right\} \cap\bigcap_{i=1}^{c} B_{i}\right)\\
    \leq\Exp{ck-(1+\epsilon)^{-\frac{1}{2}}t+c \ln N + c(2 - \ln \epsilon )}.
    \end{multline*} 
    With $t=(1+\epsilon)c\ln N$, we obtain an upper bound of the form $\exp\left(-\delta_\epsilon  c\ln N\right)$ for $N$ large enough and some $\delta_\epsilon >0$. But since $c \ln N \geq \kappa^{-1} N$ by~\eqref{eq:UB on words}, this in fact yields a bound of the form $\exp(-\delta'_\epsilon N)$ for $N$ large enough and some $\delta'_\epsilon >0$. 
    
    To conclude, we want to use a union bound over the choices of~$c$ and then the choices of~$(L_i)_{i=1}^{c}$.
   Since the number of choices is bounded by $\lambda^{-1}N$ times
    \begin{align}
     \label{eq:choices}
        \binom{N}{\lambda^{-1}N}
        <\Exp{\frac{1+\ln \lambda}{\lambda} N}
    \end{align}
    for $N$ large enough\,---\,thanks to Stirling's formula\,---, the union bound is indeed successful in showing that~\eqref{eq:ing-II} occurs with high probability, provided that $\lambda$ is large enough that $1+\ln \lambda <  \delta'_\epsilon \lambda$. 
\end{proof}

\begin{lemma}
\label{lem:ing-IV}
    For every $\epsilon > 0$,
    \[
        (N-c_N)(D_-\bar{q}(0)-\epsilon)
        \leq 
        -\sum_{i =1}^{c_N} \ln \PX[\underline{w}^{(i,N)}] 
    \]
    with high probability.
\end{lemma}

\begin{proof}[Proof sketch]
    Let $\epsilon>0$ and $\lambda\geq2$ be arbitrary. By Lemma~\ref{lem:a-priori}, we can assume the length bounds~\eqref{eq:UB on words}--\eqref{eq:LB on words}. We use the notation from Lemma~\ref{lem:ing-II}. Note that, by a union bound over the possible choices for $c$ and $(L_i)_{i=1}^c$ under the constraint that $L_{i+1}-L_i-1\geq \lambda$, it suffices to provide an exponential upper bound, made uniform in the choice of $c$, on the events
    \[
    \operatorname{Prob}\left\{-\sum_{i=1}^{c} \ln\PX[Y_{L_i}^{L_{i+1}-2}]\leq (N-c)(D_-\bar q(0)-\epsilon)\right\}.
    \]
    Note that the moment-generating function of the logarithm of the marginal on the right-hand side of~\eqref{eq:meas} is precisely $\Exp{q_{L_{i+1}-L_i -1}(\alpha)}$; {\em cf.}~\eqref{eq:qn-def}.
    Hence, using~\eqref{eq:UD} and a Chebyshev-like inequality for the function $t\mapsto\Exp{-\alpha t}$ with $\alpha < 0$,
    \begin{multline*}
        \operatorname{Prob}\left\{-\sum_{i=1}^{c} \ln\PX[Y_{L_i}^{L_{i+1}-2}]\leq (N-c)(D_-\bar q(0)-\epsilon)\right\}\\
         \leq\Exp{ck-\alpha (N-c)(D_-\bar q(0)-\epsilon)+ \sum_{i=1}^c q_{L_{i+1}-L_i-1}(\alpha)}.
    \end{multline*} 
    By convexity of the pressure~$\bar{q}$, we can take $\alpha<0$ small enough that 
    $
    \bar q(\alpha)< \alpha(D_-\bar q(0)-\tfrac \epsilon 4)
    $;
    see Figure~\ref{fig:q-and-h}. Now, pick $\lambda$ large enough that~\eqref{eq:qbar-def} and~\eqref{eq:LB on words} guarantee
    \[
    \frac{q_{L_{i+1}-L_i-1}(\alpha)}{L_{i+1}-L_i-1}<\bar q(\alpha)-\frac{\alpha\epsilon}4
    \]
    for each $i=1,\dots,c$.
    Recalling that $\sum_{i=1}^c L_{i+1}-L_i-1=N-c$ and that~\eqref{eq:LB on words} provides the bound~$c\leq\lambda^{-1}N$, we can conclude as in Lemma~\ref{lem:ing-II} using~\eqref{eq:choices} with~$\lambda$ large enough.
\end{proof}

Let $\epsilon \in (0,1)$ be arbitrary. Combining the last conclusion of Lemma~\ref{lem:a-priori} with $\lambda > (k+1)\epsilon^{-1}$ and the BC lemma, we deduce that both $c_N$ and $kc_N$ are almost surely eventually bounded by $N\epsilon$. Hence, by Lemmas~\ref{lem:cons-KSET},~\ref{lem:ing-I} and \ref{lem:ing-III} and the BC lemma,
\begin{align*}
    \limsup_{N\to\infty} \frac{c_N \ln N}{N-c_N} \leq \left(\frac{1}{1-\epsilon}\right)^2 \left(\Sc(\PY|\PX)  + \epsilon\right),
\end{align*}
almost surely. By Lemmas~\ref{lem:left-der},~\ref{lem:ing-II} and~\ref{lem:ing-IV} and the BC lemma, 
\begin{align*}
    \liminf_{N\to\infty} \frac{c_N\ln N}{N-c_N} \geq \frac{1}{1+\epsilon} (\Sc(\PY|\PX)-\epsilon),
\end{align*}
almost surely.
To deduce Theorem~\ref{thm:ieee-main}, take $\epsilon \to 0$ along some sequence and use countable additivity of probability.

\section{Numerical experiments}
\label{sec:num}

\textcolor{black}{In Figure~\ref{fig:hmm_plot2}, we compute the longest-match length estimator, the original ZM estimator and the mZM estimator for some pair of HMPs on~$\cA=\{0, 1\}$. Then, to estimate the root-mean-square error (RMSE), we compare the estimator to~\eqref{eq:cross-SMB} at $N = 2^{20}$\,---\,which is only possible in this case as we know the process.}

\begin{figure}[h]
    \centering\includegraphics{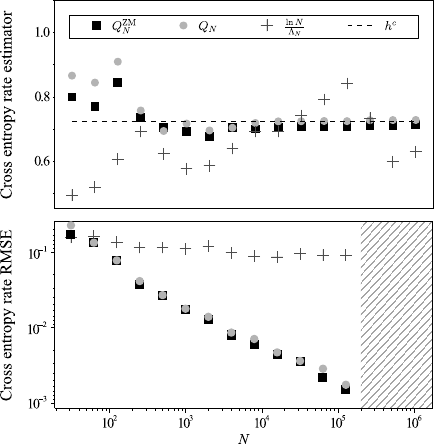}
    \caption{ \color{black} For single realizations of HMPs $\bf{X}$ and $\bf{Y}$, we plot different estimators of $\Sc(\PY|\PX)$ up to $N = 2^{20}$ (Top). Over 32 realizations, we plot our estimation of the RMSE of the different estimators up to $N = 2^{17}$ (Bottom).}
    \label{fig:hmm_plot2}
\end{figure}

 While our proofs provide no rate of convergence, these numerical experiments\,---\,and others~\cite{BGPRb}\,---\,suggest a relatively rapid \textcolor{black}{polynomial} convergence.
 Importantly, we see a significant performance advantage in the ZM-type estimators compared to the more routine longest-match length estimator. 

\section{Conclusion}

We have proved strong consistency of a slight modification of the ZM estimator of cross entropy rates between HMPs.
We emphasize that while our numerical experiments suggest little to no impact on concrete performance on HMPs, the modification of the ZM estimator allows for proofs at a level of generality that was previously inaccessible~\cite{MZ93,BGPR}.

\paragraph*{Acknowledgements} {The work of NB and RR was partially funded by the \emph{Fonds de recherche du Qu\'ebec\,---\,Nature et technologies} (FRQNT) and by the Natural Sciences and Engineering Research Council of Canada (NSERC). Part of the work of RG was partially funded by the {Rubin Gruber Science Undergraduate Research Award} and {Axel W Hundemer}. The work of GP was supported by the CY Initiative of Excellence through the grant Investissements d'Avenir ANR-16-IDEX-0008, and was done under the auspices of the \emph{Gruppo Nazionale di Fisica Matematica} (GNFM) section of the \emph{Istituto Nazionale di Alta Matematica} (INdAM) while GP was a post-doctoral researcher at University of Milano-Bicocca (Milan, Italy). }

\end{document}